\documentclass[final,3p,times,twocolumn]{elsarticle}

\usepackage{amssymb,amsmath,amsthm,bm}
\usepackage{geometry}
\usepackage{algorithm}
\usepackage{algorithmic}
\usepackage{graphicx}
\usepackage[retainorgcmds]{IEEEtrantools}
\graphicspath{{figures-pdf/}}

\usepackage[normalem]{ulem}
\usepackage{url}

\usepackage{lineno}
\usepackage{bm}

\newlength\imagewidth
\newlength\figwidth
\journal{Elsevier}
\usepackage{color}
\definecolor{dgreen}{rgb}{0,.6,0}

\newtheorem{Proposition}{Proposition}
\newtheorem{Property}{Property}

\newtheorem{Fact}{Fact}

\newcommand{\diag}{\mathop{\mathrm{diag}}}

\begin{document}

\begin{frontmatter}

\title{Cryptanalyzing a Class of Image Encryption Schemes Based on Chinese Remainder Theorem}

\author[cn-xtu-cie]{Chengqing Li\corref{corr}}
\ead{chengqingg@gmail.com}

\author[cn-xtu-cie]{Yuansheng Liu}

\author[hk-cityu]{Leo Yu Zhang}

\author[hk-cityu]{Kwok-Wo Wong}

\cortext[corr]{Corresponding author.}

\address[cn-xtu-cie]{MOE (Ministry of Education) Key Laboratory of Intelligent Computing and Information Processing, \\
College of Information Engineering,
Xiangtan University, Xiangtan 411105, Hunan, China}


\address[hk-cityu]{Department of Electronic Engineering, City University of Hong Kong, Hong Kong}

\begin{abstract}
As a fundamental theorem in number theory, the Chinese Reminder Theorem (CRT) is widely used to construct cryptographic primitives.
This paper investigates the security of a class of image encryption schemes based on CRT, referred to as CECRT. Making use of
some properties of CRT, the equivalent secret key of CECRT can be recovered efficiently. The required number of pairs of chosen
plaintext and the corresponding ciphertext is only $(1+\lceil (\log_2L)/l \rceil)$. The attack complexity is only $O(L)$,
where $L$ is the plaintext length and $l$ is the number of bits representing
a plaintext symbol. In addition, other defects of CECRT such as invalid compression function and low sensitivity to plaintext, are reported.
The work in this paper will help clarify positive role of CRT in cryptology.
\end{abstract}

\end{frontmatter}

\section{Introduction}
\label{sec:intro}

Both the transmission and the storage of digital data have dual requirements of high
operating efficiency and security, which lead to the joint operations of compression and encryption.
According to the order of the operations, joint compression and encryption schemes can
be categorized into three classes: encryption on compressed data \cite{Shujun:AnalysiBS:JEI06,Wu:Integrate:IEEETM:2005,Chen:IS:2010}; simultaneous compression and encryption \cite{Zhou:ISPL:2007,Zhou:IEEETCASI:2008,Jakimoski:Analysis:IEEEMultimedia08,Li:BreakingSecLZW:ICME2011,ChenJY:Joint:TCSII11}; compression on encrypted data \cite{chang:cryptanalysis:2002,Zhang:ITIFS:2011,Klinc:CompressAES:IEEEIT2012}. Recently, Chinese Remainder Theorem (CRT) is used in constructing simultaneous compression and encryption schemes or the basis of some efficient encryption algorithms.

The earliest known example of CRT can be found in the book, \textit{The Mathematical Classic of Sunzi},
written by Chinese mathematician Sun Tzu in the fifth century. In 1247, another Chinese mathematician Jiushao Qin
generalized it into a statement about simultaneous congruences and provided the complete solution in \textit{Mathematical Treatise in Nine Sections}
\cite{SHEN:AHES:1988}. Antiquity of Chinese mathematicians' study on the remainder problem (and maybe sparsity of Chinese mathematicians' contribution
to classic mathematics) made the complete form of the
statement be called Chinese Remainder Theorem. As a fundamental theorem in number theory, it has been widely used in
various fields of information security, e.g. speed up implementation of the RSA algorithm \cite{Yen:ITC:2003,Wang:ITSP:2010}, secret sharing \cite{Goldreich:ITIT:2000}, and secure code \cite{Ling:CodeCRT:IEEETIT2001}. For a comprehensive survey of the cryptographic applications of CRT and
chaos-based cryptanalysis, please refer to \cite{Ding:CRT:96} and \cite{Li:Rules:IJBC2006}, respectively.

As reviewed in \cite[Sec.~4.3.2]{Knuth:ACP:1997}, CRT supports the modular representation of a large number (dividend) as a set of numbers (remainders) in some given small domains. It converts the addition, subtraction, and multiplication of large numbers into very simple operations on small numbers. In addition, the conversion provides simultaneous operations on different moduli for parallel computing. Considering these benefits, a number of symmetric encryption schemes based on CRT have been proposed since 2001. The schemes designed in \cite{ammar:CRT:IEEE2001,Wang:crt:IWSRTA04,Vikram:CRT:ICSPCN07,Chang:CRT:ISJ10} all consider the gray level of some plain-image pixels as remainders, and the summing divisor of CRT as the cipher-element, where the moduli sequences are considered as the secret key or key stream.
Conversely, the scheme proposed in \cite{Meher:IISCSVP:2006} combines the gray levels of some plain-image pixels into a big divisor and stores the smaller remainder as the cipher-elements. Reference \cite{Aithal:JMMB:2010} follows this idea and further encrypts the remainders in a stream cipher mode, using two pseudo-random number sequences (PRNS). In 2013, an image encryption scheme, called CECRT in this paper, was proposed \cite{Zhu:CRT:SPIC13}. It first permutes the pixels of the plain-image and then performs the CRT operations as reported in \cite{ammar:CRT:IEEE2001,Wang:crt:IWSRTA04,Vikram:CRT:ICSPCN07,Chang:CRT:ISJ10}.
In \cite{Vikram:CRT:ICSPCN07,Zhu:CRT:SPIC13}, the authors claimed that their schemes possess the feature of simultaneous compression and encryption.
Since cryptanalysis is an integral work to evaluate security level of any encryption scheme \cite{Solak:IS:2011}, it is important to analyze security properties
of the encryption schemes based on CRT.

As CECRT is a typical example of the class of symmetric encryption schemes based on CRT and almost
all security defects of other schemes can be found in it, we will focus on breaking CECRT. We
found a property of CRT on the relationship among the product of some moduli, the divisor corresponding to a special set of remainders, and the divisor.
To the best of our knowledge, this is the first time that the property of CRT is reported. Based on it, we prove that the diffusion part of CECRT can be compromised efficiently using only a pair of chosen-plaintext
and the corresponding ciphertext. Then, the permutation part of CECRT can be broken using the existing standard cryptanalysis methods. In addition, the following security defects of CECRT are also reported:
1) the compression performance of CECRT is marginal and even negative; 2) the ciphertext is not sensitive to changes in the plaintext;
3) the moduli of CRT are not suitable to be used as a sub-key.

The rest of this paper is organized as follows. In Sec.~\ref{sec:CECRT}, CECRT is briefly described.
Then, the comprehensive cryptanalyses on CECRT are presented in Sec.~\ref{sec:cryptanalysis}, together with detailed experimental results.
The last section concludes the paper.

\section{Description of CECRT}
\label{sec:CECRT}

The kernel of CECRT relies on the Chinese Remainder Theorem,
which states that the system of linear congruences
\begin{equation}
\left\{x\equiv q_i\pmod{m_i}\right\}_{i=1}^t
\label{eq:congruence}
\end{equation}
has unique solution
\begin{equation}
x \equiv \sum_{i=1}^{t} e_i \tilde{m}_i q_i \pmod{m},
\label{eq:solutionCRT}
\end{equation}
when $m_1, m_2, \ldots, m_t$ are coprime integers,
where $\tilde{m}_i=m/m_i$, $m=\prod_{i=1}^tm_i$, $(e_i\tilde{m}_i)\equiv 1\pmod{m_i}$,
$\{q_i\}_{i=1}^t\subset \mathbb{Z}$, and $t$ is an integer larger than or equal to one. Let $\mathbf{P}=\{p_i\}_{i=1}^L$ and $\mathbf{C}=\{c_i\}_{i=1}^{L/k}$ denote the plaintext
and the corresponding ciphertext, respectively, where $k$ is the number of plaintext symbols encrypted at one time. Without loss of generality,
$L$ is assumed to be a multiple of $k$. Then, the basic operations of CECRT are described
as follows\footnote{For the sake of completeness, some notations in the original paper \cite{Zhu:CRT:SPIC13}
are modified provided that the essential form of CECRT is not changed.}.
\begin{itemize}
\item \textit{The secret key} consists of $k$ coprime integers, represented as an ordered set $\mathbf{N}=(n_i)_{i=1}^k$, and
$(x_0, y_0,$ $a_1, a_2, b_1, b_2, b_3)$, where $n_i \geq 256$, $(x_0, y_0)$ and $(a_1, a_2, b_1, b_2, b_3)$ are, respectively, the
initial condition and the control parameters of the $2$D hyper-chaotic system
\begin{equation}
\label{2dChaosSystem}
\begin{cases}
x_{n+1}=a_1x_n+a_2y_n,\\
y_{n+1}=b_1+b_2x_n^{2}+b_3y_n,\\
\end{cases}
\end{equation}
.

\item \textit{The initialization process:}
1) Iterate the chaotic system~(\ref{2dChaosSystem}) for $500+L$ times from the
initial condition $(x_0, y_0)$ with the control parameters
$(a_1, a_2, b_1, b_2, b_3)$ and obtain two PRNS $\mathbf{X} = \{x_i\}_{i=1}^{L}$ and
$\mathbf{Y}=\{y_i\}_{i=1}^{L}$ after discarded the first $500$ states;
2) Sort $\mathbf{X}$ and $\mathbf{Y}$ in ascending order, then derive two intermediate permutation sequences
$\mathbf{U}=\{u(i)\}_{i=1}^{L}$ and $\mathbf{V}=\{v(i)\}_{i=1}^{L}$ by comparing $\mathbf{X}$ and
$\mathbf{Y}$ with their sorted versions, respectively, where $x_{u(i)}$ denotes the
$i$-th smallest element of $\mathbf{X}$ and $y_{v(i)}$ denotes the
$i$-th smallest element of $\mathbf{Y}$;
3) Combine the two vectors $\mathbf{U}$ and $\mathbf{V}$ and
obtain a \textit{permutation relation vector} (a bijective map on the entities of a plaintext) $\mathbf{W}=\{w(i)\}_{i=1}^{L}$, where $w(i)=v(u(i))$.

\item \textit{The encryption process} is comprised of the following two basic operations:

\textit{1) Permutation:} For $i=1\sim L$, set
\begin{equation}
\nonumber
h_i = p_{w(i)}.
\end{equation}
\textit{2) Confusion:} For $j=1\sim L/k$, set $c_j$ as
the solution of the system of linear congruences $\left\{x\equiv h_{(j-1)k+i}\pmod{n_i}\right\}_{i=1}^k$,
namely
\begin{equation}
\label{equiv:encrypt}
c_j = \sum_{i=1}^{k} e_i  \tilde{n}_i \cdot p_{w((j-1)k+i)}  \bmod{n},
\end{equation}
where $\tilde{n}_i=(\prod_{j=1}^kn_j)/n_i$, and
\begin{equation}
(e_i\tilde{n}_i)\equiv 1\pmod{n_i}.
\label{eq:condition}
\end{equation}

\item \textit{The decryption process} consists of two steps:

\textit{1) Inverse Confusion:} For $i=1\sim L$, set
\begin{equation}
h_{i}=c_{(\lfloor (i-1)/k \rfloor + 1)}\bmod{n_{((i-1)\bmod k + 1)}}.
\label{eq:invertCRT}
\end{equation}

\textit{2) Inverse Permutation:}
For $i=1\sim L$, set
\begin{equation*}
p_{i}=h_{w^{-1}(i)},
\end{equation*}
where $\mathbf{W}^{-1}=\{w^{-1}(i)\}_{i=1}^{L}$ is the inverse of $\mathbf{W}$.
\end{itemize}

\section{Cryptanalysis}
\label{sec:cryptanalysis}

To carry out an efficient chosen-plaintext attack on CECRT,
some properties of Chinese Remainder Theorem are introduced first.

\subsection{Properties of Chinese Remainder Theorem}
\label{subsec:property}

\begin{Property}
\label{proposition:crtProperty}
Given a set $\{s_i\}_{i=1}^r\subset \{1, 2, \cdots, t\}$, one has
\begin{equation}
\prod_{i=1}^r m_{s_i}=\gcd\left(\sum_{i=1}^{r} (e_{s_i} \tilde{m}_{s_i})-1, m \right)
\label{eq:gcdproperty11}
\end{equation}
and
\begin{equation}
\prod_{i=1}^{t-r} m_{t_i}=\gcd\left(\sum_{i=1}^r(e_{s_i} \tilde{m}_{s_i}), \prod_{i=1}^{t-r} m_{t_i}\right),
\label{eq:gcdproperty12}
\end{equation}
where $\{t_i\}_{i=1}^{t-r}=\{1, 2, \cdots, t\}-\{s_i\}_{i=1}^r$.
\end{Property}
\begin{proof}
Given $i\in \{1, \cdots, r\}$,  $(e_{s_i} \tilde{m}_{s_i})\equiv 1\pmod{m_{s_i}}$,
one can obtain
\begin{equation}
\gcd(e_{s_i} \tilde{m}_{s_i}-1, m_{s_i})=m_{s_i}.
\label{eq:gcdms}
\end{equation}
For any $j\in \{1, \cdots, r\}$ and $j\neq i$,
$\gcd(e_{s_j} \tilde{m}_{s_j}, m_{s_i})=m_{s_i}$. Therefore, one gets
\begin{equation}
\gcd\left(\sum_{j=1, j\neq i}^r(e_{s_j} \tilde{m}_{s_j}), m_{s_i}\right)=m_{s_i}.
\label{eq:gcdmsi}
\end{equation}
Combining Eq.~(\ref{eq:gcdms}) and Eq.~(\ref{eq:gcdmsi}), the result is
\begin{equation*}
\gcd\left(\sum_{j=1}^r(e_{s_j} \tilde{m}_{s_j})-1, m_{s_i}\right)=m_{s_i}.
\label{eq:gcdmsij}
\end{equation*}
Referring to Fact~\ref{proposition:gcdone}, it can be derived that
\begin{IEEEeqnarray}{rCl}
\prod_{i=1}^rm_{s_i} & = & \prod_{i=1}^r\gcd\left(\sum_{j=1}^r(e_{s_j} \tilde{m}_{s_j})-1, m_{s_i}\right) \nonumber\\
                     & = & \gcd\left(\sum_{i=1}^r(e_{s_i} \tilde{m}_{s_i})-1, \prod_{i=1}^rm_{s_i}\right) \label{eq:gcdmsijone}.
\end{IEEEeqnarray}
As
\begin{equation*}
\sum_{j=1}^r\left(e_{s_j} \tilde{m}_{s_j}\right)=\sum_{j=1}^r\left(e_{s_j} \cdot \frac{\prod_{i=1}^r m_{s_i}}{m_{s_j}} \cdot \frac{m}{\prod_{i=1}^rm_{s_i}}\right),
\end{equation*}
one can get
\begin{equation*}
\frac{m}{\prod_{i=1}^rm_{s_i}}=\gcd\left(\sum_{j=1}^r(e_{s_j} \tilde{m}_{s_j}),\frac{m}{\prod_{i=1}^rm_{s_i}}\right),
\end{equation*}
which is the same as Eq.~(\ref{eq:gcdproperty12}). Referring to Proposition~\ref{proposition:gcd}, one can obtain
\begin{equation}
\gcd\left(\sum_{i=1}^r(e_{s_i} \tilde{m}_{s_i})-1,\frac{m}{\prod_{i=1}^rm_{s_i}}\right)=1
\label{eq:gcdone}
\end{equation}
from the above equation.
Multiplying the two sides of Eq.~(\ref{eq:gcdmsijone}) and Eq.~(\ref{eq:gcdone}), respectively, one has
\begin{IEEEeqnarray*}{rCl}
\prod_{i=1}^rm_{s_i} & = & \gcd\left(\sum_{i=1}^r(e_{s_i} \tilde{m}_{s_i})-1, \prod_{i=1}^rm_{s_i}\right)\\
                     &   & {\ } \cdot\gcd\left(\sum_{i=1}^r(e_{s_i} \tilde{m}_{s_i})-1,\frac{m}{\prod_{i=1}^rm_{s_i}}\right)\\
                     & = & \gcd\left(\sum_{i=1}^{r} (e_{s_i} \tilde{m}_{s_i})-1, m \right).
\end{IEEEeqnarray*}
\end{proof}

\begin{Fact}
\label{proposition:gcdone}
Given three integers $a$, $b$ and $c$. If $\gcd(b, c) =1$, then
$$\gcd(a, bc)= \gcd(a, b)\cdot \gcd(a, c).$$
\end{Fact}

\begin{Proposition}
\label{proposition:gcd}
Given two integers $a$, $b$ and $|a|+|b|\neq 0$.
If $\gcd(a, b)=b$, then $\gcd(a-1, b)=1$.
\end{Proposition}
\begin{proof}
Since $\gcd(a, b)=b$, there exists an integer $k$ such that $a = k\cdot b$.
Then, $\gcd(a-1, b)=\gcd(k \cdot b-1,b)=\gcd((k-1)\cdot b + (b-1), b) =\gcd(b-1, b) = \gcd(b-1, (b-1)+1) =\gcd(b-1, 1) = 1$.
\end{proof}

\begin{Property}
The coefficients of CRT, $\{e_i\}_{i=1}^t$, $\{\tilde{m}_i\}_{i=1}^t$, satisfy
\begin{equation}
\sum_{i=1}^{t} e_i \tilde{m}_i\equiv 1\pmod{m}.
\label{eq:gcdproperty2}
\end{equation}
\end{Property}
\begin{proof}
In the linear congruences shown in Eq.~(\ref{eq:congruence}), assume $q_i\equiv 1$. Obviously,
$x=1$. As the solution is unique, one can obtain Eq.~(\ref{eq:gcdproperty2}) by setting the values of $x$ and $\{q_i\}_{i=1}^t$ in Eq.~(\ref{eq:solutionCRT}).
\end{proof}

\begin{Property}
Equation~(\ref{eq:congruence}) and Eq.~(\ref{eq:solutionCRT}) determine a pair of reciprocal bijective map between $\mathbb{Z}_m$ and $\{(q_1, q_2, \cdots, q_t)\}$, where
$q_i$ goes through $\mathbb{Z}_{m_i}$ for $i=1\sim t$.
\end{Property}
\begin{proof}
Obviously, Eq.~(\ref{eq:congruence}) determines a bijective map between $\mathbb{Z}_m$ and $\{(q_1, q_2, \cdots, q_t)\}$.
As the solution of Eq.~(\ref{eq:solutionCRT}) is unique, one can assure that the map determined by Eq.~(\ref{eq:solutionCRT}) is the
reciprocal of that confirmed by Eq.~(\ref{eq:congruence}).
\end{proof}

\subsection{Chosen-plaintext attack}
\label{subsec:cpa}

\subsubsection{Determining the modulus $n=\prod_{i=1}^kn_i$}

Assume the ciphertext symbols distribute uniformly, one has
\begin{equation*}
Prob\left(\max\left(\{c_i\}_{i=1}^{L/k}\right)\neq (n-1)\right)=\left(1-1/n\right)^{L/k}.
\end{equation*}
As $f(n)=\left(1-1/n\right)^{n}$ is a monotonic increasing function, $0.31<f(n)<\lim\limits_{n\rightarrow \infty}\left(1-1/n\right)^{n}=e^{-1}$ when $n> 3$.
In general, $(L/k)<<n$, one has $0.31<f(n)<\left(1-1/n\right)^{L/k}$. As a result, the modulus $n$ is not equal to $\max\left(\{c_i\}_{i=1}^{L/k}\right)+1$ with a negligible probability.

The value of the modulus $n$ can be guessed from the above approximated value and then be verified by checking its coincidence with the values of its factors obtained in the forthcoming attack. Besides this, one can also recover it from the ciphertext of a chosen-plaintext or a known-plaintext of special format such as a binary image. If $\mathbf{P}=\{p_i\}_{i=1}^L\subset\{0, 1\}$, then one can assure $1\le S\le 2^k$ from Eq.~(\ref{equiv:encrypt}), where
$S$ denotes the cardinality of $\mathbb{B}$, the set containing different numbers in the corresponding ciphertext. Add up each pair of the elements in $\mathbb{B}$ and put all the results into an array $\mathbb{\widehat{B}}$ of upper limit size $\binom{2^k}{2}=2^{k-1}(2^k-1)$.
From Property~3, one can see that there is at most one element in $\mathbb{\widehat{B}}$ which is equal to 1. By Property~2, one can assure that the element in $\mathbb{B}$ possessing the highest frequency is $n+1$ since
about $S/2-1$ elements in $\mathbb{\widehat{B}}$ are equal to $(n+1)\in (1, 2n-2]$. Other elements exist at a probability much lower than $(S/2-1)/\binom{S}{2}=\frac{S-2}{S(S-1)}$ when $S$ approaches $2^k$. Therefore, the value of $n$ can be easily found using one of the three methods described above.

\subsubsection{Recovering the unordered set $\{n_i\}_{i=1}^k$}

Once the value of $n$ has been confirmed, the elements in $\mathbf{N}$ may be recovered
by factoring it. However, the computational complexity is extremely high, especially when $n$ is very large. Fortunately,
this can be performed efficiently by comparing the above chosen-plaintext and its ciphertext. As for the
binary plaintext, one can obtain the following Proposition from Eq.~(\ref{equiv:encrypt}) and Property~1.
\begin{Proposition}
When the plaintext $\mathbf{P}=\{p_i\}_{i=1}^L\subset\{0, 1\}$, its corresponding ciphertext satisfies
\begin{equation*}
\gcd(c_j-1, n)=\prod_{i=1}^r n_{s_i}
\end{equation*}
and
\begin{equation*}
\gcd(c_j, n)=\left(\frac{n}{\prod_{i=1}^r n_{s_i}}\right)\cdot d,
\end{equation*}
where
\begin{equation}
\{s_i\}_{i=1}^r=\{i \mid \ p_{w((j-1)k+i)}=1, i=1\sim k\},
\label{eq:set1}
\end{equation}
$d$ divides $n$ exactly, and $j=1\sim L/k$.
\end{Proposition}

From Proposition~2, one can see that all the elements of $\widehat{\mathbf{N}}$ come from the elements in $\{n_i\}_{i=1}^k$ or multiple of their products (maybe themselves), where
\[\widehat{\mathbf{N}}=\{\gcd(c_j-1, n), \gcd(c_j, n), n/\gcd(c_j-1, n)\}_{j=1}^{L/k}.\]
Making use of the properties of $\widehat{\mathbf{N}}$, $\{n_i\}_{i=1}^k$
can be recovered by the following steps:
\begin{itemize}
\item \textit{Step 1}: Set $\widetilde{\mathbf{N}}$ with the smallest numbers in $\widehat{\mathbf{N}}$ which are co-prime with each other.
If the cardinality of $\widetilde{\mathbf{N}}$ is equal to $k$, one can assure that $\widetilde{\mathbf{N}}=\{n_i\}_{i=1}^k$ and stops the search.

\item \textit{Step 2}: For any two elements of $\widehat{\mathbf{N}}$, add the greatest common divisor (gcd) of them to $\widehat{\mathbf{N}}$
if the gcd is not equal to one.

\item \textit{Step 3}: For any two elements of $\widehat{\mathbf{N}}$, add the quotient of them into $\widehat{\mathbf{N}}$
if one can be divided with no remainder by another. Go to \textit{Step 1}.
\end{itemize}

Obviously, the above attack can be carried out in the same way if any two chosen-plaintexts, whose difference is a binary text,
and the corresponding ciphertexts are available.

\subsubsection{Breaking the permutation part of CECRT}

From Proposition~\ref{prop:equikey}, one can see that $\widehat{\mathbf{N}}$ is equivalent to $\mathbf{N}$ in terms of the
function on decrypting CECRT. This means that CECRT is reduced to a permutation-only encryption scheme once the unordered set $\widehat{\mathbf{N}}$
has been recovered. By the general cryptanalysis method based on multi-branch tree proposed in \cite{Lcq:Optimal:SP11}, the equivalent permutation part of CECRT can be revealed from $\lceil (\log_2L)/l \rceil$ pairs of chosen-plaintext, where $l$ is the number of bits representing
a plaintext symbol \cite{Li:AttackingPOMC2008}. Note that the binary chosen-plaintext used in
the above sub-section can also be employed to verify some permutation relations utilizing Eq.~(\ref{eq:set1}).
\begin{Proposition}
The order of the elements in set $\{n_i\}_{i=1}^k$
has no influence on the decryption of CECRT.
\label{prop:equikey}
\end{Proposition}
\begin{proof}
Once the set $\{n_i\}_{i=1}^k$ has been determined, one can get an approximate version of $\mathbf{N}$ (denoted as $\mathbf{N}^*$), with
elements of $\mathbf{N}$ in any order, namely,
\begin{equation*}
\mathbf{N}^* =\mathbf{N}\cdot \mathbf{T},
\end{equation*}
where $\mathbf{T}$ is a permutation matrix of size $k\times k$.
From Eq.~(\ref{eq:invertCRT}), the approximate of $\{h_i\}_{i=1}^k$ corresponding to $ \mathbf{T}$
can be calculated by
\begin{align}
\left(h^*_{(j-1)\cdot k+i}\right)_{i=1}^k & = \left(h_{(j-1)\cdot k+i}\right)_{i=1}^k\cdot \mathbf{T} \nonumber \\
                                          & = \left(p_{w((j-1)\cdot k+i)}\right)_{i=1}^k\cdot \mathbf{T},
\label{eq:permuted}
\end{align}
where $j=1\sim L/k$.
Denote $\mathbf{\widetilde{W}}$ by
\begin{equation*}
\mathbf{\widetilde{W}}=\mathbf{W}\cdot \mathbf{\widehat{T}},
\end{equation*}
where $\mathbf{\widehat{T}}=\diag(\mathbf{T},\cdots, \mathbf{T})$ is a permutation matrix of size
$L\times L$, whose main diagonal blocks are all $\mathbf{T}$. Obviously, one has
\begin{IEEEeqnarray}{rCl}
\mathbf{\widetilde{W}}^{-1} & = & \mathbf{W}^{-1} \cdot \mathbf{\widehat{T}}^{-1}\nonumber\\
                            & = & \mathbf{W}^{-1} \cdot \diag(\mathbf{T}^{-1},\cdots, \mathbf{T}^{-1}).\label{eq:falsepermuter}
\end{IEEEeqnarray}
Combining Eq.~(\ref{eq:permuted}) and Eq.~(\ref{eq:falsepermuter}), one can assure that the influence of $\mathbf{T}$ on the decryption of CECRT
is elliminated.
\end{proof}

\subsubsection{Analyzing performance of the breaking method}

The probability that the sub-key for confusion in CECRT can be exactly recovered depends on
whether $\widehat{\mathbf{N}}$ contains independent information of $n_i$, $i=1\sim k$. It is
difficult to work out the exact probability. However, by a large number of random experiments,
we found that this probability approaches one when the percentage of non-zero elements in the available chosen plaintext
is not too small. It is easy to verify that the complexity of recovering $\{n_i\}_{i=1}^k$ is $O(L)$.

Although the sub-key used in the permutation part is difficult to be determined, its equivalent
version $\mathbf{W}$ can be fully revealed by some chosen plaintexts. It can also be recovered
at a very high accuracy even if only some known plaintexts are available \cite{Li:AttackingPOMC2008}.
Utilizing the data structure of multi-branch tree, intersection of the multiple permutation relations generated by
the known plaintexts can be converted into linear operation of visit. So, the complexity required in recovering equivalent version of the
permutation part is $O(L)$ \cite{Lcq:Optimal:SP11}.

\subsubsection{Verifying the breaking method with experimental results}

To verify the effectiveness of the proposed chosen-plaintext attack,
a large number of experiments have been performed using plain-images of various sizes.
Here a typical example is shown. In this experiment, the secret parameters of CECRT are chosen
as follows: $(n_1, n_2, n_3, n_4)=(311, 313, 317, 293)$, $(x_0, y_0)=(0.0394, 0.001)$
and $(a_1, a_2, b_1,b_2, b_3)=(-0.95, $ $-1.3, -0.45, 2.4, 1.05)$, which is the key configuration used in \cite{Zhu:CRT:SPIC13}.
The available information for breaking the confusion sub-key in CECRT includes the binary image ``Bricks'' of size $512\times 512$ and the corresponding cipher-image. They are depicted in Fig.~\ref{figure:encryptpeppers}, where four consecutive pixels are used to denote one cipher-element, and the width is the same as that of the plain-image.
\begin{figure}[!htb]
\centering
\begin{minipage}{\figwidth}
\centering
\includegraphics[width=\textwidth]{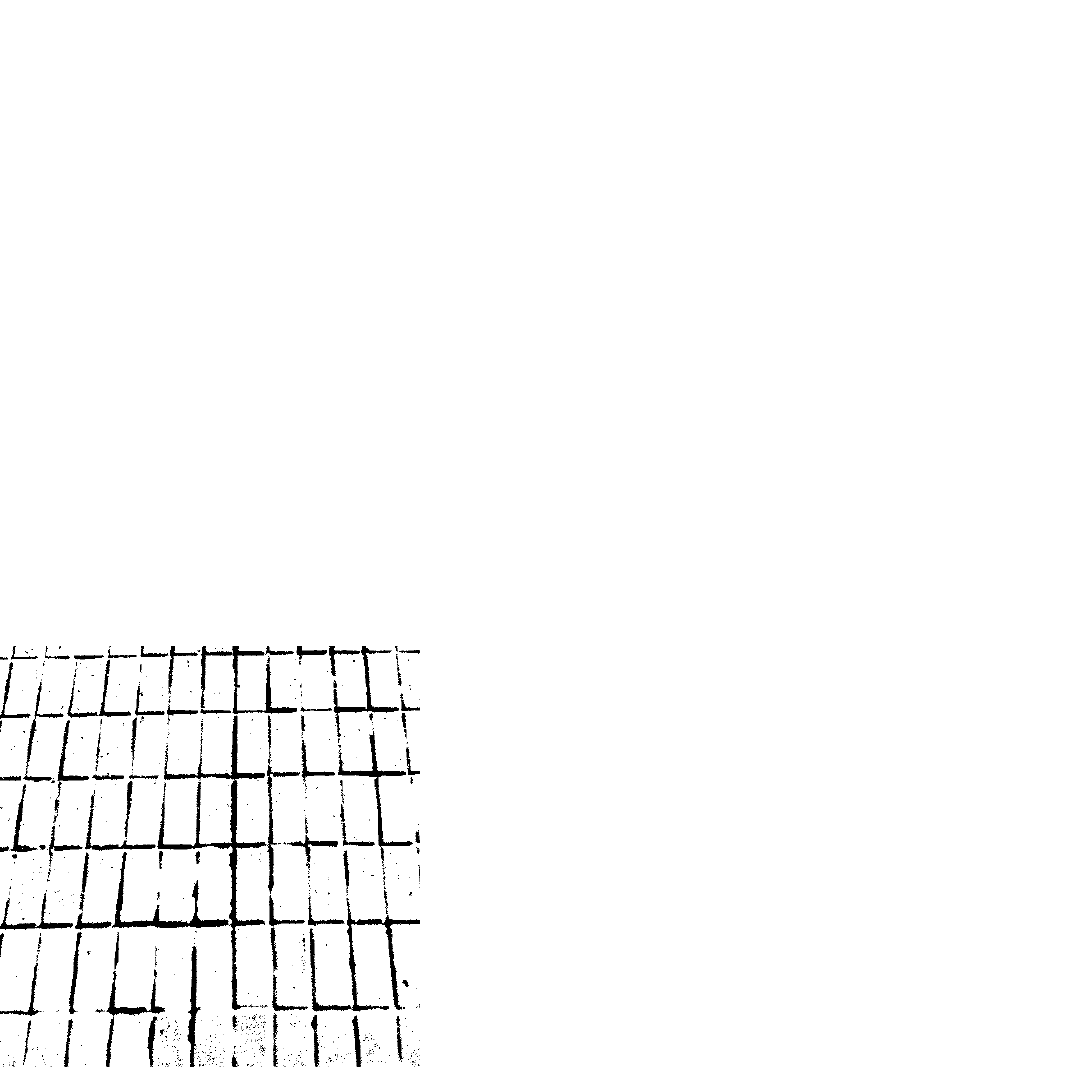}
a)
\end{minipage}
\begin{minipage}{\figwidth}
\centering
\includegraphics[width=\textwidth]{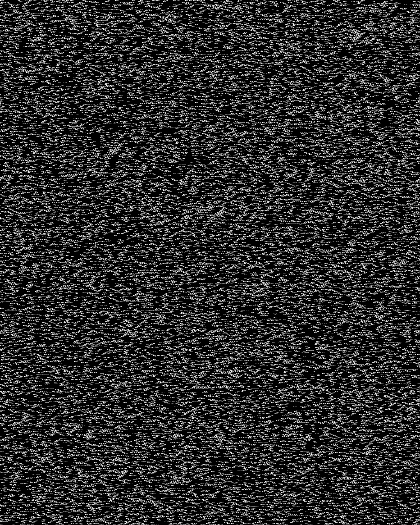}
b)
\end{minipage}
\caption{The binary plain-image ``Bricks'' and the corresponding
cipher-image: a) the plain-image; b) the cipher-image.}
\label{figure:encryptpeppers}
\end{figure}

As shown in Fig.~\ref{figure:distribution}a), distribution of the elements in $\mathbb{\widehat{B}}$ is not
uniform and one element exists at a much higher probability than others, which agrees with the above
analysis. To further verify this, distribution of the elements in $\mathbb{\widehat{B}}$ when $k=6, 8, 10$
are also plotted in Fig.~\ref{figure:distribution}, where $\{n_i\}_{i=1}^6=(419, 323, 649, 501, 302, 449\}$, $\{n_i\}_{i=1}^8=(573, 593,$ $443, 577, 341, 428, 293, 541)$, $\{n_i\}_{i=1}^{10}=$ $(323,$ $273, 263,$ $349, 625, 409, 436, 451, 389, 479)$, respectively. As a result, $n=311\cdot313\cdot317\cdot293=9041315183$ is easily obtained. Then, one can get $\widehat{\mathbf{N}}=\{30857731, 91709, 98587, 29071753, 28885991, 28521499, 928$ $81,91123, 99221, 311, 97343, 317, 313, 293\}$.
In this case, $\{n_i\}_{i=1}^4$ can be recovered in the first search step. Finally, $\lceil \log_2(512\cdot512)/8\rceil=3$ chosen plain-images are employed to recover the permutation relation vector. As the detailed results of general cryptanalysis problems have been presented in \cite{Li:AttackTDCEA2005,Li:AttackingPOMC2008,Lcq:Optimal:SP11}, the related experimental results are omitted here.
\begin{figure}[!htb]
\centering
\begin{minipage}{2\figwidth}
\centering
\includegraphics[width=\textwidth]{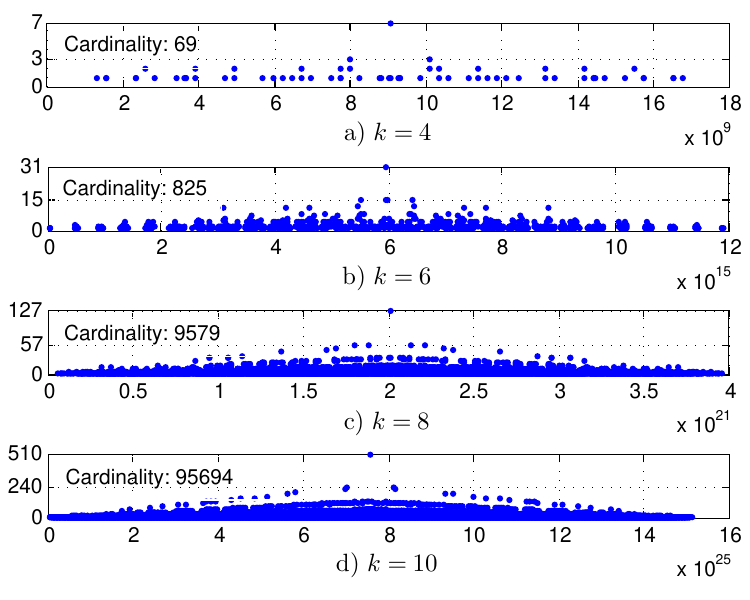}
\end{minipage}
\caption{Distribution of the elements in $\mathbb{\widehat{B}}$ when $k=4, 6, 8, 10$.}
\label{figure:distribution}
\end{figure}

\subsection{Other defects of CECRT}
\label{subsec:defect}

Obviously, CECRT encrypts any plaintext of fixed value as ciphertext having a constant value.
In particular, this encryption scheme fails to encrypt the fixed plaintext having only zero value. Besides these, CECRT suffers from the following defects.
\begin{itemize}
\item \textit{Invalid compression of CECRT}
By Property~3 of CRT, one can find that the function defined by CRT is bijective.
In fact, every lossless compression algorithm is a bijective function between the message
being encoded and the codeword. The size of the message domain
is substantially reduced by extracting the redundancy exist in the original message. However, there is no
such operation in CECRT. Refer to Eq.~(\ref{equiv:encrypt}), one can calculate the ratio between the bit length of the $j$-th cipher-element and that of
the corresponding plain-element:
\begin{IEEEeqnarray*}{rCl}
\frac{\lceil \log_2n \rceil}{\sum_{i=1}^k \lceil\log_2p_{w((j-1)k+i))} \rceil} & \ge &
\frac{\lceil \sum_{i=1}^k \log_2n_i \rceil}{ \sum_{i=1}^k\lceil \log_2(n_i)\rceil}\\
                 & \ge & \frac{\sum_{i=1}^k\lceil \log_2n_i\rceil-(k-1)}{ \sum_{i=1}^k\lceil \log_2n_i\rceil }\\
                 &= & 1-\frac{k-1}{ \sum_{i=1}^k\lceil \log_2n_i\rceil }.
\end{IEEEeqnarray*}
As $n_i\ge 256$, the above ratio is greater than or equal to $\frac{7+(1/k)}{8}$, which means that the compression performance of CECRT is marginal. Even worse,
it may produce a ciphertext which is longer than the plaintext. For the cipher-image shown in Fig.~\ref{figure:encryptpeppers}b), one can find that the expansion ratio is $\frac{\lceil\log_2(311)\rceil+\lceil\log_2(313)\rceil+\lceil\log_2(317)\rceil+\lceil\log_2(293)\rceil}{4\cdot 8}
=\frac{9+9+9+9}{32}=36/32=9/8$.

\item \textit{Low sensitivity with respect to change of plaintext}

A secure cryptographic algorithm is expected to possess the \textit{avalanche effect}, i.e., a tiny change in plaintext will cause each bit
of the ciphertext change at a probability of one half. Unfortunately, CECRT fails to have this desired property.
When two plaintexts $\mathbf{P}=\{p_i\}_{i=1}^L$ and $\mathbf{P}^*=\{p^*_i\}_{i=1}^L$ satisfying $p_{w((j-1)k+i)}-p^*_{w((j-1)k+i)}\equiv d$ for
one $j$, one can assure that the difference of the corresponding cipher-elements is $(c_j-c^*_j)=\sum_{i=1}^{k} (e_i \tilde{n}_i) d \bmod{n}=d\bmod n$.
The low sensitivity with respect to a general change of plaintext can be supported by the findings in \cite{Li:RobustCRT:IEEETSP2009}, which proved that the
recovery of a large integer from its remainders using CRT is not sensitive to the change (error) in the remainders.

\item \textit{Improper usage of CRT moduli as sub-key}

As mentioned in \cite{Knuth:ACP:1997}, the moduli of CRT can be simply set as some powers of 2 minus one when
they are fixed parameters. However, the configuration become much complex when they are dynamic.
This is attributed to the fact that the probability that any $k$ positive integers are relatively prime is given by
\[A_k=\prod_{p \rm\;prime}\left(1-\frac{1}{p}\right)^{k-1}\left(1+\frac{k-1}{p}\right)\]
\cite{Toth:ProbPrime:FQ2002}. The value of $A_k$ decreases exponentially with respect to
$k$, where $A_3\approx 0.286$, $A_8\approx0.001$, and $A_{10}<10^{-4}$. From Eq.~(\ref{equiv:encrypt}), one can see that
CECRT of a given secret key generates at most $2^k$ different cipher-elements. When the plaintext is gray-scale image, the value of $k$ should not be less than 8. Therefore, an efficient scheme for generating
a set of coprime integers under the control of a secret key has to be found. Otherwise, it would cost a large amount of computation
to search and verify eligible sub-keys. In addition, a change in $n_i$ only influences the decryption of $1/k$ plain-elements, which
hardly meets the requirement of a secure encryption scheme, e.g. a small change in the key should cause a drastic
change in the reconstructed plaintext.
\end{itemize}

\section{Conclusion}

The security of a class of encryption schemes using the Chinese Remainder Theorem has been analyzed in detail.
Based on some properties of CRT, the sub-key used in the confusion part of CECRT can be easily recovered with only one pair of chosen plaintext and the corresponding ciphertext. Then, the permutation part can be broken with the reported general method.
In addition, other defects of CECRT are reported to facilitate the proper use of
CRT in cryptography. The work in this paper may be extended to analyze the security of other cryptographic
applications also based on CRT.

\section*{Acknowledgement}

This research was supported by the National Natural Science Foundation of China (No.~61100216), and the Alexander von Humboldt Foundation of Germany.

\bibliographystyle{elsarticle-num}
\bibliography{CRT}

\begin{thebibliography}{10}
\expandafter\ifx\csname url\endcsname\relax
  \def\url#1{\texttt{#1}}\fi
\expandafter\ifx\csname urlprefix\endcsname\relax\def\urlprefix{URL }\fi
\expandafter\ifx\csname href\endcsname\relax
  \def\href#1#2{#2} \def\path#1{#1}\fi

\bibitem{Shujun:AnalysiBS:JEI06}
S.~Li, C.~Li, K.-T. Lo, G.~Chen, Cryptanalysis of an image encryption scheme,
  Journal of Electronic Imaging 15~(4) (2006) 043012.

\bibitem{Wu:Integrate:IEEETM:2005}
C.-P. Wu, C.-C.~J. Kuo, Design of integrated multimedia compression and
  encryption systems, IEEE Transactions on Multimedia 7~(5) (2005) 828--839.

\bibitem{Chen:IS:2010}
T.-H. Chen, C.-S. Wu, Compression-unimpaired batch-image encryption combining
  vector quantization and index compression, Information Sciences 180~(9)
  (2010) 1690--1701.

\bibitem{Zhou:ISPL:2007}
J.~Zhou, Z.~Liang, Y.~Chen, O.~C. Au, Security analysis of multimedia
  encryption schemes based on multiple huffman table, IEEE Signal Processing
  Letters 14~(3) (2007) 201--204.

\bibitem{Zhou:IEEETCASI:2008}
J.~Zhou, O.~C. Au, Comments on ``a novel compression and encryption scheme
  using variable model arithmetic coding and coupled chaotic system", IEEE
  Transactions on Circuits and Systems I-Regular Papers 55~(10) (2008)
  3368--3369.

\bibitem{Jakimoski:Analysis:IEEEMultimedia08}
G.~Jakimoski, K.~Subbalakshmi, Cryptanalysis of some multimedia encryption
  schemes, IEEE Transactions on Multimedia 10~(3) (2008) 330--338.

\bibitem{Li:BreakingSecLZW:ICME2011}
S.~Li, C.~Li, C.-C.~J. Kuo, On the security of a secure {Lempel-Ziv-Welch
  (LZW)} algorithm, in: Proceedings of 2011 IEEE International Conference on
  Multimedia and Expo, 2011.

\bibitem{ChenJY:Joint:TCSII11}
J.~Chen, J.~Zhou, K.-W. Wong, A modified chaos-based joint compression and
  encryption scheme, IEEE Transactions on Circuits and Systems II 58~(2) (2011)
  110--114.

\bibitem{chang:cryptanalysis:2002}
C.-C. Chang, T.-X. Yu, Cryptanalysis of an encryption scheme for binary images,
  Pattern Recognition Letters 23~(14) (2002) 1847--1852.

\bibitem{Zhang:ITIFS:2011}
X.~Zhang, Lossy compression and iterative reconstruction for encrypted image,
  IEEE Transactions on Information Forensics and Security 6~(1) (2011) 53--58.

\bibitem{Klinc:CompressAES:IEEEIT2012}
D.~Klinc, C.~Hazay, A.~Jagmohan, H.~Krawczyk, T.~Rabin, On compression of data
  encrypted with block ciphers, IEEE Transactions on Information Theory 58~(11)
  (2012) 6989--7001.

\bibitem{SHEN:AHES:1988}
K.~Shen, Historical development of the {C}hinese remainder theorem, Archive for
  History of Exact Sciences 38~(4) (1988) 285--305.

\bibitem{Yen:ITC:2003}
S.-M. Yen, S.~Kim, S.~Lim, S.-J. Moon, {RSA} speedup with {C}hinese remainder
  theorem immune against hardware fault cryptanalysis, IEEE Transactions on
  Computers 52~(4) (2003) 461--472.

\bibitem{Wang:ITSP:2010}
W.~Wang, X.-G. Xia, A closed-form robust {C}hinese remainder theorem and its
  performance analysis, IEEE Transactions on Signal Processing 58~(11) (2010)
  5655--5666.

\bibitem{Goldreich:ITIT:2000}
O.~Goldreich, D.~Ron, M.~Sudan, Chinese remaindering with errors, IEEE
  Transactions on Information Theory 46~(4) (2000) 1330--1338.

\bibitem{Ling:CodeCRT:IEEETIT2001}
S.~Ling, P.~Sole, On the algebraic structure of quasi-cyclic codes i: Finite
  fields, IEEE Transactions on Information Theory 47~(7) (2001) 2751--2760.

\bibitem{Ding:CRT:96}
C.~Ding, D.~Pei, A.~Salomaa, Chinese Remainder Theorem: Applications in
  Computing, Coding, Cryptography, World Scientific Publishing, 1996.

\bibitem{Li:Rules:IJBC2006}
G.~\'{A}lvarez, S.~Li, Some basic cryptographic requirements for chaos-based
  cryptosystems, International Journal of Bifurcation and Chaos 16~(8) (2006)
  2129--2151.

\bibitem{Knuth:ACP:1997}
D.~E. Knuth, The Art of Computer Programming Vol. 2: Seminumerical Algorithms,
  3rd Edition, Addison-Wesley, 1997.

\bibitem{ammar:CRT:IEEE2001}
A.~Ammar, A.~A. Kabbany, M.~Youssef, A.~Amam, A secure image coding scheme
  using residue number system, in: Proceedings of the Eighteenth National Radio
  Science Conference, Vol.~2, IEEE, 2001, pp. 399--405.

\bibitem{Wang:crt:IWSRTA04}
W.~Wang, M.~Swamy, M.~Ahmad, {RNS} application for digital image processing,
  in: Proceedings of 4th IEEE International Workshop on System-on-Chip for
  Real-Time Applications, 2004, pp. 77--80.

\bibitem{Vikram:CRT:ICSPCN07}
V.~Jagannathan, A.~Mahadevan, R.~Hariharan, S.~Srinivasan, Number theory based
  image compression encryption and application to image multiplexing, in:
  Proceedings of IEEE International Conference on Signal Processing,
  Communications and Networking, 2007, pp. 59--64.

\bibitem{Chang:CRT:ISJ10}
J.~Yang, C.~Chang, C.~Lin, Residue number system oriented image encoding
  schemes, The Imaging Science Journal 58~(1) (2010) 3--11.

\bibitem{Meher:IISCSVP:2006}
P.~K. Meher, J.~C. Patra, A new approach to secure distributed storage, sharing
  and dissemination of digital image, in: Proceedings of IEEE International
  Symposium on Circuits and Systems, 2006, pp. 373--376.

\bibitem{Aithal:JMMB:2010}
G.~Aithal, K.~N.~H. Bhat, U.~S. Acharya, High-speed and secure encryption
  schemes based on chinese remainder theorem for storage and transmission of
  medical information, Journal of Mechanics in Medicine and Biology 10~(1)
  (2010) 167--190.

\bibitem{Zhu:CRT:SPIC13}
H.~Zhu, C.~Zhao, X.~Zhang, A novel image encryption-compression scheme using
  hyper-chaos and {C}hinese remainder theorem, Signal Processing-Image
  Communication 28~(6) (2013) 670--680.

\bibitem{Solak:IS:2011}
E.~Solak, C.~Cokal, Algebraic break of image ciphers based on discretized
  chaotic map lattices, Information Sciences 181~(1) (2011) 227--233.

\bibitem{Lcq:Optimal:SP11}
C.~Li, K.-T. Lo, Optimal quantitative cryptanalysis of permutation-only
  multimedia ciphers against plaintext attacks, Signal Processing 91~(4) (2011)
  949--954.

\bibitem{Li:AttackingPOMC2008}
S.~Li, C.~Li, G.~Chen, N.~G. Bourbakis, K.-T. Lo, A general quantitative
  cryptanalysis of permutation-only multimedia ciphers against plaintext
  attacks, Signal Processing: Image Communication 23~(3) (2008) 212--223.

\bibitem{Li:AttackTDCEA2005}
C.~Li, S.~Li, G.~Chen, G.~Chen, L.~Hu, Cryptanalysis of a new signal security
  system for multimedia data transmission, EURASIP J. on Applied Signal
  Processing 2005~(8) (2005) 1277--1288.

\bibitem{Li:RobustCRT:IEEETSP2009}
X.~Li, H.~Liang, X.-G. Xia, A robust chinese remainder theorem with its
  applications in frequency estimation from undersampled waveforms, IEEE
  Transactions on Signal Processing 57~(11) (2009) 4314--4322.

\bibitem{Toth:ProbPrime:FQ2002}
L.~Toth, The probability that $k$ positive integers are pairwise relatively
  prime, Fibonacci Quarterly 40~(1) (2002) 13--18.

\end{thebibliography}
\end{document}